\documentclass[aps,pra,reprint,superscriptaddress]{revtex4-2}
\usepackage{amsmath}
\usepackage{amssymb}
\usepackage{amsfonts}
\usepackage{amsthm}
\usepackage{mathrsfs}
\usepackage{bm}
\usepackage{array}
\usepackage{booktabs}
\usepackage{graphicx}
\usepackage{xcolor}
\usepackage[colorlinks=true,linkcolor=red,citecolor=blue,urlcolor=blue]{hyperref}
\usepackage{braket}
\usepackage{orcidlink}

\newtheorem{theorem}{Theorem}
\newtheorem{proposition}{Proposition}

\theoremstyle{remark}
\newtheorem{remark}{Remark}
\theoremstyle{plain}

\newcommand{\Nph}{\mathcal{N}_{\mathrm{phase}}}
\newcommand{\Ndep}{\mathcal{N}_{\mathrm{dep}}}
\newcommand{\Nad}{\mathcal{N}_{\mathrm{AD}}}
\newcommand{\Nid}{\mathcal{N}_{0}}

\begin{document}

\title{Analytic Benchmarks for Coherence-to-Entanglement Conversion
  under Post-Gate Noise in CNOT-Based Protocols}

\author{Asad Ali\orcidlink{0000-0001-9243-417X}}
\email{asal68826@hbku.edu.qa}
\affiliation{Qatar Center for Quantum Computing, College of Science and Engineering,
  Hamad Bin Khalifa University, Doha, Qatar}

\author{H.~Kuniyil\orcidlink{0000-0003-0338-1278}}
\affiliation{Qatar Center for Quantum Computing, College of Science and Engineering,
  Hamad Bin Khalifa University, Doha, Qatar}

\author{M.~I~Hussain\orcidlink{0000-0002-6231-7746}}
\affiliation{Qatar Center for Quantum Computing, College of Science and Engineering, Hamad Bin Khalifa University, Doha, Qatar}

\author{M.~T.~Rahim\orcidlink{0000-0003-1529-928X}} 
\affiliation{Qatar Center for Quantum Computing, College of Science and Engineering, Hamad Bin Khalifa University, Doha, Qatar}

\author{A.~Slaoui\orcidlink{0000-0002-5284-3240}}
\affiliation{LPHE-Modelling and Simulation, Faculty of Sciences, Muhammad V University
  in Rabat, Rabat, Morocco}
\affiliation{Centre of Physics and Mathematics, CPM, Faculty of Sciences, Muhammad V
  University in Rabat, Rabat, Morocco}

\author{Saif Al-Kuwari\orcidlink{0000-0002-4402-7710}}
\email{smalkuwari@hbku.edu.qa}
\affiliation{Qatar Center for Quantum Computing, College of Science and Engineering,
  Hamad Bin Khalifa University, Doha, Qatar}

\date{\today}

\begin{abstract}
Coherence-to-entanglement conversion transforms single-qubit superposition into a practical two-qubit resource, but noise limits this process in near-term quantum hardware. We derive closed-form benchmarks for a minimal CNOT primitive in which a coherent qubit and an incoherent ancilla generate entanglement before undergoing phase damping, global depolarizing, amplitude damping, or independent local depolarizing noise. Using the $\ell_1$-norm of coherence and negativity, we prove the noiseless law $\mathcal{N}_0=C_{\ell_1}/2$, valid for arbitrary mixed inputs, and obtain exact negativities, survival fractions, and entanglement-sudden-death thresholds. For all $X$-state-preserving channels, a master relation shows that entanglement loss results from the competition between coherence suppression and partial-transpose spectral shifts. Phase damping yields $\eta=1-p$ without finite-noise sudden death; global depolarization gives coherence-dependent sudden death; amplitude damping adds an excited-population penalty and sudden death only for $\theta>\pi/4$; while local depolarization is most destructive at equal depolarizing strength. The initial survival slopes, $-1$, $-3/2$, $-2$, and $-3$, act as compact noise fingerprints. Since concurrence satisfies $C=2\mathcal{N}$ for the generated states, all robustness rankings remain unchanged. Mapping channel parameters to $T_1$, $T_\varphi$, and average gate fidelity connects the theory to hardware-level performance.
\end{abstract}


\maketitle

\section{Introduction}
\label{sec:intro}

Quantum coherence and entanglement are two fundamental and interconvertible
resources in quantum information science~\cite{nielsen2010quantum,
horodecki2009quantum,baumgratz2013quantifying,streltsov2017colloquium,
naseri2022entanglement,shi2017coherence,feng2023coherence,liu2019coherence,
ahnefeld2022coherence,ye2026coherence}.
The resource theory of coherence, formalized by Baumgratz, Cramer, and
Plenio~\cite{baumgratz2013quantifying}, established the $\ell_1$-norm of coherence
as an operationally motivated monotone and showed that local coherence can be
converted into bipartite entanglement by coupling a coherent system to an incoherent
ancilla via an incoherent operation~\cite{streltsov2015measuring,streltsov2017colloquium}.
This convertibility underpins the resource-theoretic view of coherence as an
input that can be transformed into entanglement.

A canonical example of such a conversion is the CNOT gate.  Preparing qubit $A$ in
the coherent state $\ket{\psi_A}=\cos\theta\ket{0}+e^{i\phi}\sin\theta\ket{1}$ and
qubit $B$ in the incoherent reference state $\ket{0}$, the CNOT gate produces
$\ket{\Psi}_{AB}=\cos\theta\ket{00}+e^{i\phi}\sin\theta\ket{11}$, converting local
single-qubit coherence into bipartite entanglement.  For $\theta=\pi/4$ the output
is a Bell state; for $\theta=0$ or $\pi/2$ the input is incoherent and no
entanglement is generated.  In realistic quantum devices, however, this
conversion is degraded by environmental noise---dephasing, energy relaxation, and
depolarization~\cite{clerk2010introduction,suter2016colloquium}---acting
after the gate.  Understanding how different noise mechanisms affect
coherence-to-entanglement conversion is essential for characterising and improving
quantum processors.

This paper studies coherence-to-entanglement conversion in this minimal protocol by
treating \emph{four} post-gate channels on an equal footing---phase damping, global
depolarizing, amplitude damping, and the \emph{independent local single-qubit
depolarizing} channel, the last of which yields the most fragile behaviour of the
four---and by organizing all four channels through a single competition between
coherence suppression and partial-transpose spectrum shift (Section~\ref{sec:ESD}).
The specific contributions of the present work are as follows.
\begin{enumerate}
  \item We derive closed-form entanglement-survival fraction
    $\eta(\theta,p)=\mathcal{N}/\mathcal{N}_0$ under \emph{four} noise channels
    and compare them analytically and graphically.
  \item We show that the four channels produce qualitatively distinct degradation
    patterns: (i)~multiplicative suppression with no sudden death (phase damping);
    (ii)~coherence-dependent sudden death for all nonzero inputs, with more coherent
    inputs more robust (global depolarizing); (iii)~population-transfer correction
    causing sudden death only for $\theta>\pi/4$, with the robustness ordering
    \emph{reversed} relative to depolarizing (amplitude damping); and (iv)~the
    strongest degradation, with sudden death for all inputs at the lowest noise
    threshold (local single-qubit depolarizing).
  \item We show that all four channels are governed by a single master relation
    $\mathcal{N}=\max[0,\,f(\text{noise})\,\tfrac{1}{2}\sin2\theta-g(\text{noise})]$,
    reducing the four-channel comparison to a per-channel pair of
    coherence-suppression and spectrum-shift functions (Table~\ref{tab:fg}), and
    derive the analytic sudden-death thresholds as explicit functions of the input
    coherence angle.
  \item We establish that all results are \emph{measure-independent}: every state
    in the protocol satisfies $C=2\mathcal{N}$ exactly (concurrence equals twice the
    negativity), so all thresholds and orderings are identical under either measure,
    and the noiseless identity $\mathcal{N}_0=\tfrac{1}{2}C_{\ell_1}$ holds for
    arbitrary mixed single-qubit inputs, not only pure ones.
  \item We connect the noise parameters to physical device quantities ($T_1$, $T_2$,
    gate fidelity) and provide order-of-magnitude estimates for superconducting
    processors.
\end{enumerate}

The CNOT gate throughout this paper is assumed \emph{ideal}; all noise is applied
after the gate as an independent post-gate channel.  This deliberate idealisation
allows each channel's effect to be isolated analytically.  The limitations of this
model are discussed explicitly in Sections~\ref{sec:physical_params}
and~\ref{sec:conclusion}.

The paper is organised as follows.  Section~\ref{sec:framework} defines the input
state, $\ell_1$-norm coherence, and negativity.  Section~\ref{sec:unitary} derives the
noiseless $\mathcal{N}_0=\tfrac{1}{2}C_{\ell_1}$ relation.
Section~\ref{sec:channels} defines the four noise channels.
Section~\ref{sec:PT} develops the $X$-state partial-transpose framework.
Sections~\ref{sec:phase}, \ref{sec:dep}, \ref{sec:AD}, and~\ref{sec:loc} derive
the exact negativities under phase damping, global depolarizing, amplitude
damping, and local depolarizing, respectively.
Section~\ref{sec:ESD} analyses entanglement sudden death.
Section~\ref{sec:efficiency} compares entanglement-survival fractions.
Section~\ref{sec:physical_params} maps noise parameters to physical quantities.
Section~\ref{sec:results_summary} collects all closed-form results.
Discussion and conclusions are in Sections~\ref{sec:discussion}
and~\ref{sec:conclusion}.

\section{Theoretical Framework}
\label{sec:framework}

\subsection{System and Initial State}

We consider a two-qubit system with joint Hilbert space
$\mathcal{H}_{AB}\cong\mathbb{C}^4$, using the computational basis
$\{\ket{00},\ket{01},\ket{10},\ket{11}\}$.
Qubit $A$ is initialised in
\begin{equation}
  \ket{\psi_A}=\cos\theta\,\ket{0}+e^{i\phi}\sin\theta\,\ket{1},
  \quad \theta\in\!\left[0,\tfrac{\pi}{2}\right],\quad\phi\in[0,2\pi),
  \label{eq:psiA}
\end{equation}
and qubit $B$ in the incoherent state $\rho_B=\ket{0}\!\bra{0}$.

\subsection{Quantum Coherence Measure}

The $\ell_1$-norm of coherence of qubit $A$, defined with respect to the computational
basis~\cite{baumgratz2013quantifying}, is
\begin{equation}
  C_{\ell_1}(\rho_A)=\sum_{i\neq j}|\rho_{ij}|.
  \label{eq:Cl1_def}
\end{equation}
For the pure state~\eqref{eq:psiA}:
\begin{equation}
  C_{\ell_1}(\ket{\psi_A})=\sin(2\theta).
  \label{eq:Cl1}
\end{equation}
The coherence is maximal at $\theta=\pi/4$ and vanishes at $\theta=0,\pi/2$.

\textbf{Role of the phase $\phi$.}
All negativity and survival-fraction expressions derived below depend only on
$|c|=\tfrac{1}{2}\sin(2\theta)$, the magnitude of the CNOT output's
anti-diagonal coherence, which is independent of $\phi$.  All results therefore
hold for arbitrary $\phi\in[0,2\pi)$.

\subsection{Entanglement Measure}

Bipartite entanglement is quantified by the negativity~\cite{vidal2003g,
plenio2005introduction}:
\begin{equation}
  \mathcal{N}(\rho_{AB})=\sum_{\lambda_i<0}|\lambda_i|,
  \label{eq:neg_evals}
\end{equation}
where $\{\lambda_i\}$ are the eigenvalues of the partial transpose $\rho_{AB}^{T_B}$.
For two-qubit systems, $\mathcal{N}>0$ is equivalent to entanglement.
We use the negativity throughout for its analytic transparency, but the choice
carries no loss of generality: every state produced by the protocol below (in the
noiseless limit and under all four channels) is an $X$-state for which the Wootters
concurrence~\cite{wootters1998entanglement} satisfies $C=2\mathcal{N}$ exactly
(Appendix~\ref{app:calc}).  Consequently all thresholds, survival fractions, and
robustness orderings reported here are \emph{measure-independent}: replacing
negativity by concurrence rescales every quantity by a factor of two and leaves the
sudden-death boundaries and orderings unchanged.

\section{Ideal Conversion Protocol}
\label{sec:unitary}

Applying the CNOT gate to $\ket{\psi_A}\otimes\ket{0}$ yields
\begin{equation}
  \ket{\Psi}=\cos\theta\,\ket{00}+e^{i\phi}\sin\theta\,\ket{11}.
  \label{eq:Psi}
\end{equation}
The corresponding density operator is
\begin{equation}
  \rho=
  \begin{pmatrix}
    a & 0 & 0 & c \\ 0 & 0 & 0 & 0 \\ 0 & 0 & 0 & 0 \\ c^* & 0 & 0 & b
  \end{pmatrix},
  \label{eq:rho_Xstate}
\end{equation}
where $a=\cos^2\theta$, $b=\sin^2\theta$, $c=\tfrac{1}{2}e^{-i\phi}\sin(2\theta)$
(see Appendix~\ref{app:calc} for the full derivation).

\begin{proposition}[Noiseless coherence--entanglement relation]
\label{prop:CE}
For the CNOT protocol with input~\eqref{eq:psiA}:
\begin{equation}
  \mathcal{N}_0(\theta)=\tfrac{1}{2}\sin(2\theta)=\tfrac{1}{2}C_{\ell_1}(\ket{\psi_A}).
  \label{eq:CE_relation}
\end{equation}
\end{proposition}
The proof follows from $|c|=\tfrac{1}{2}\sin(2\theta)$ and the $X$-state negativity
result of Proposition~\ref{prop:Xstate}; see Appendix~\ref{app:calc}.
The maximum $\mathcal{N}_0=1/2$ is reached at $\theta=\pi/4$.

\begin{remark}[Mixed-input extension of the noiseless relation]
\label{rem:mixed}
The identity~\eqref{eq:CE_relation} is not restricted to pure inputs.  For an
arbitrary single-qubit input
$\rho_A=\bigl(\begin{smallmatrix}\rho_{00}&\rho_{01}\\\rho_{01}^*&\rho_{11}\end{smallmatrix}\bigr)$
with incoherent ancilla, the CNOT maps the off-diagonal $\rho_{01}$ directly onto
the $\{\ket{00},\ket{11}\}$ coherence of the output while populating only the
diagonal elsewhere, so the noiseless output negativity is
$\mathcal{N}_0=|\rho_{01}|$.  Since the $\ell_1$-norm of coherence of a single
qubit is $C_{\ell_1}(\rho_A)=2|\rho_{01}|$, the relation
$\mathcal{N}_0=\tfrac{1}{2}C_{\ell_1}(\rho_A)$ holds for \emph{every} input state,
pure or mixed.  The pure-state expression~\eqref{eq:CE_relation} is the special
case $|\rho_{01}|=\tfrac{1}{2}\sin(2\theta)$.
\end{remark}

\section{Noise Channels}
\label{sec:channels}

After the ideal CNOT, the two-qubit state is subjected to a post-gate noise map.
The noise parameter in each channel is a phenomenological strength in $[0,1]$
whose physical interpretation depends on the channel; see
Section~\ref{sec:physical_params}.

\subsection{Phase-Damping Channel}

The single-qubit phase-damping channel with dephasing strength $p$ is defined by
the Kraus operators~\cite{nielsen2010quantum}
\begin{equation}
  K_0=\begin{pmatrix}1&0\\0&\sqrt{1-p}\end{pmatrix},\quad
  K_1=\begin{pmatrix}0&0\\0&\sqrt{p}\end{pmatrix}.
  \label{eq:PD_Kraus}
\end{equation}
Off-diagonal elements are suppressed by $\sqrt{1-p}$; populations are unchanged.
The two-qubit channel is $\mathcal{E}_{\mathrm{phase}}^{(2)}=
\mathcal{E}_{\mathrm{phase}}\otimes\mathcal{E}_{\mathrm{phase}}$.

\subsection{Global Depolarizing Channel}

The global two-qubit depolarizing channel is
\begin{equation}
  \mathcal{E}_{\mathrm{dep}}(\rho)=(1-p)\rho+\frac{p}{4}\mathbb{I}_4,
  \quad p\in[0,1].
  \label{eq:dep_channel}
\end{equation}
Unlike phase damping, it also shifts all diagonal populations toward $1/4$,
producing isotropic mixing.

\textbf{Global vs.\ local depolarizing.}
Equation~\eqref{eq:dep_channel} is a \emph{global} two-qubit channel, distinct from
applying independent single-qubit depolarizing channels
$\mathcal{E}_p(\rho_j)=(1-p)\rho_j+p\tfrac{1}{2}\mathbb{I}_{2}$~\cite{nielsen2010quantum} to each qubit
(where $p$ is the per-qubit depolarization probability, related to the Pauli-error
rate by $p=4q/3$).
The analytic negativity under local depolarizing is derived in
Section~\ref{sec:loc} and its landscape is shown in Fig.~\ref{fig:landscape}.
We choose the global model as the primary focus because it is the natural
error model for two-qubit gates characterised by randomised benchmarking.
For the channel $\mathcal{E}_{\mathrm{dep}}(\rho)=(1-p)\rho+p\,\mathbb{I}_4/4$ used
here, the parameter $p$ is fixed by the average gate fidelity through
$p=\tfrac{4}{3}(1-F_{\mathrm{avg}})$~\cite{nielsen2002simple,magesan2011scalable}
(derived in Section~\ref{sec:physical_params}).

\subsection{Amplitude Damping Channel}

The single-qubit amplitude damping (AD) channel models $T_1$ energy relaxation
from $\ket{1}$ to $\ket{0}$.  Its Kraus operators are~\cite{nielsen2010quantum}
\begin{equation}
  A_0=\begin{pmatrix}1&0\\0&\sqrt{1-\gamma}\end{pmatrix},\quad
  A_1=\begin{pmatrix}0&\sqrt{\gamma}\\0&0\end{pmatrix},
  \quad\gamma\in[0,1].
  \label{eq:AD_Kraus}
\end{equation}
The parameter $\gamma=1-e^{-t/T_1}$ is the decay probability.
We apply independent AD to both qubits, giving the two-qubit channel
$\mathcal{E}_{\mathrm{AD}}^{(2)}=\mathcal{E}_{\mathrm{AD}}\otimes\mathcal{E}_{\mathrm{AD}}$
with Kraus operators $A_{ij}=A_i\otimes A_j$.  Unlike phase damping, AD also
changes populations, creating a qualitatively different degradation depending on
$\theta$.

\section{Restricted X-State Structure and Partial Transpose}
\label{sec:PT}

Both phase damping and global depolarizing preserve the $X$-state block structure of
the CNOT output.  Amplitude damping preserves a general $X$-state with additional
nonzero populations.  The key algebraic tool is:

\begin{proposition}[Negativity of the restricted $X$-state]
\label{prop:Xstate}
For the restricted $X$-state
$\rho_X=\mathrm{diag}(a',0,0,b')+c'\ket{00}\!\bra{11}+c'^*\ket{11}\!\bra{00}$
with $a',b'\ge0$, $a'+b'=1$:
\begin{equation}
  \mathcal{N}(\rho_X)=|c'|.
  \label{eq:N_Xstate}
\end{equation}
\end{proposition}
The proof and the partial-transpose spectra for all four channels are given in
Appendix~\ref{app:calc}.  This result applies directly to the phase-damping output;
the depolarizing and amplitude damping outputs require separate calculations owing
to their modified diagonal entries.

\section{Phase-Damping Channel: Exact Negativity}
\label{sec:phase}

Independent phase damping suppresses the two-qubit coherence $\ket{00}\!\bra{11}$
by a factor $(1-p)$ while leaving populations unchanged.  The evolved state is
\begin{equation}
  \rho_{\mathrm{phase}}=\begin{pmatrix}
    \cos^2\theta & 0 & 0 & \frac{1-p}{2}e^{-i\phi}\sin(2\theta) \\
    0 & 0 & 0 & 0 \\ 0 & 0 & 0 & 0 \\
    \frac{1-p}{2}e^{i\phi}\sin(2\theta) & 0 & 0 & \sin^2\theta
  \end{pmatrix},
  \label{eq:rho_phase}
\end{equation}
which retains the restricted $X$-state form.
Proposition~\ref{prop:Xstate} then gives directly:
\begin{equation}
  \Nph(\theta,p)=\tfrac{1}{2}(1-p)\sin(2\theta).
  \label{eq:Nphase}
\end{equation}
Eq.~\eqref{eq:Nphase} is the first main result; it is listed in Table~\ref{tab:results_summary}.
The phase-damping negativity is a purely multiplicative suppression,
$\Nph=(1-p)\mathcal{N}_0$.  Entanglement survives for all $p<1$ and any
$\theta\in(0,\pi/2)$: \emph{no finite-noise sudden death}.
The two-dimensional landscape $\Nph(\theta,p)$ is shown in
Fig.~\ref{fig:landscape}(a).
Numerical validation (Appendix~\ref{app:numerics}) confirms agreement to
$\Delta_{\max}<10^{-12}$.

\section{Global Depolarizing Channel: Exact Negativity}
\label{sec:dep}

Applying $\mathcal{E}_{\mathrm{dep}}$ to the CNOT output yields (see
Appendix~\ref{app:calc}).
Defining $\alpha=(1-p)\cos^2\!\theta+\tfrac{p}{4}$,
$\beta=(1-p)\sin^2\!\theta+\tfrac{p}{4}$,
$d=\tfrac{p}{4}$, and $c=\tfrac{1-p}{2}e^{-i\phi}\sin(2\theta)$:
\begin{equation}
  \rho_{\mathrm{dep}}=
  \begin{pmatrix}
    \alpha & 0 & 0 & c \\
    0 & d & 0 & 0 \\
    0 & 0 & d & 0 \\
    c^* & 0 & 0 & \beta
  \end{pmatrix}.
\end{equation}
The partial transpose has eigenvalues $p/4\pm\tfrac{1}{2}(1-p)\sin(2\theta)$
in the $\{\ket{01},\ket{10}\}$ block (Appendix~\ref{app:calc}), giving:
\begin{equation}
  \Ndep(\theta,p)=\max\!\left[0,\,\tfrac{1}{2}(1-p)\sin(2\theta)-\tfrac{p}{4}\right].
  \label{eq:Ndep}
\end{equation}
At $p=0$ this reduces to $\mathcal{N}_0$.  The isotropic identity admixture
shifts the partial-transpose spectrum by $p/4$, allowing the negative eigenvalue
to cross zero at a finite noise strength.
The two-dimensional landscape $\Ndep(\theta,p)$ is shown in
Fig.~\ref{fig:landscape}(b), where the white dashed curve marks $p_c(\theta)$.
Numerical validation confirms $\Delta_{\max}<10^{-12}$ (Appendix~\ref{app:numerics}).

\section{Amplitude Damping Channel: Exact Negativity}
\label{sec:AD}

\subsection{Evolved Density Matrix}

Applying independent amplitude damping to both qubits (Kraus operators
$A_{ij}=A_i\otimes A_j$, $i,j\in\{0,1\}$) to the CNOT output~\eqref{eq:rho_Xstate}
yields (derivation in Appendix~\ref{app:calc}):
Setting $a=\cos^2\!\theta+\gamma^2\sin^2\!\theta$,
$b=(1-\gamma)^2\sin^2\!\theta$,
$f=\gamma(1-\gamma)\sin^2\!\theta$, and
$c=\tfrac{1}{2}(1-\gamma)\sin(2\theta)\,e^{-i\phi}$:
\begin{equation}
  \rho_{\mathrm{AD}}=
  \begin{pmatrix}
    a & 0 & 0 & c \\
    0 & f & 0 & 0 \\
    0 & 0 & f & 0 \\
    c^* & 0 & 0 & b
  \end{pmatrix}.
  \label{eq:rho_AD}
\end{equation}
This is a general $X$-state.  The nonzero $\ket{01}$/$\ket{10}$ populations
$\gamma(1-\gamma)\sin^2\!\theta$ arise from transitions $\ket{11}\to\ket{01}$
and $\ket{11}\to\ket{10}$ via amplitude damping of individual qubits.

\subsection{Partial Transpose and Spectrum}

Taking $T_B$ maps $\ket{00}\!\bra{11}\to\ket{01}\!\bra{10}$, giving
Using the same abbreviations $a,b,f,c$ as Eq.~\eqref{eq:rho_AD}:
\begin{equation}
  \rho_{\mathrm{AD}}^{T_B}=
  \begin{pmatrix}
    a & 0 & 0 & 0 \\
    0 & f & c & 0 \\
    0 & c^* & f & 0 \\
    0 & 0 & 0 & b
  \end{pmatrix}.
  \label{eq:PT_AD}
\end{equation}
This is block-diagonal.  The $2\times2$ central block in the
$\{\ket{01},\ket{10}\}$ subspace has eigenvalues
\begin{equation}
  \lambda_\pm=(1-\gamma)\!\left[\gamma\sin^2\!\theta\pm\tfrac{1}{2}\sin(2\theta)\right].
  \label{eq:evals_AD}
\end{equation}
The negative eigenvalue is
\begin{align}
  \lambda_- &=(1-\gamma)\!\left[\gamma\sin^2\!\theta-\tfrac{1}{2}\sin(2\theta)\right]\nonumber\\
            &=(1-\gamma)\sin\theta\!\left[\gamma\sin\theta-\cos\theta\right],
\end{align}
which is negative if and only if $\gamma<\cot\theta$.

\subsection{Exact Negativity under Amplitude Damping}

The negativity is the absolute value of the single negative partial-transpose
eigenvalue~\eqref{eq:evals_AD}, $\Nad=\max[0,-\lambda_-]=\max[0,|\lambda_-|]$.
Substituting $\lambda_-=(1-\gamma)[\gamma\sin^2\!\theta-\tfrac{1}{2}\sin(2\theta)]$
from the spectrum above gives the closed form
\begin{equation}
  \Nad(\theta,\gamma)
  =\max\!\left[0,\,(1-\gamma)\!\left(\tfrac{1}{2}\sin(2\theta)-\gamma\sin^2\!\theta\right)\right].
  \label{eq:Nad}
\end{equation}

\textbf{Physical interpretation.}
Equation~\eqref{eq:Nad} has a two-term structure:
\begin{equation}
  \Nad=(1-\gamma)\Nid-\gamma(1-\gamma)\sin^2\!\theta.
  \label{eq:Nad_decomposed}
\end{equation}
The first term $(1-\gamma)\Nid$ mirrors the phase-damping result with $p\to\gamma$:
a multiplicative coherence suppression.  The second term $\gamma(1-\gamma)\sin^2\!\theta$
is a \emph{population-transfer correction} unique to amplitude damping, arising
because $\ket{11}\to\ket{01}$ and $\ket{11}\to\ket{10}$ transitions populate
the central partial-transpose block and shift $\lambda_-$ upward.  This correction
is proportional to $\sin^2\!\theta$, the excited-state population of the CNOT output.
At $\gamma=0$ one recovers $\mathcal{N}_0$; at $\gamma=1$ the negativity vanishes.
The full two-dimensional landscape $\Nad(\theta,\gamma)$ is shown in
Fig.~\ref{fig:landscape}(c), where the ESD boundary
$\gamma_c(\theta)=\cot\theta$ (white dashed curve, $\theta>\pi/4$ only)
separates the entangled and separable regions.
Numerical validation (Appendix~\ref{app:numerics}) confirms $\Delta_{\max}<10^{-12}$.

\section{Local Depolarizing Channel: Exact Negativity}
\label{sec:loc}

\subsection{Evolved Density Matrix}

Applying independent single-qubit depolarizing
$\mathcal{E}_p(\rho_j)=(1-p)\rho_j+p\tfrac{1}{2}\mathbb{I}_{2}$
to both output qubits of the CNOT yields an $X$-state with
anti-diagonal element $(1-p)^2c$ and equal off-diagonal populations
$\beta=p(2-p)/4$ in the $\{\ket{01},\ket{10}\}$ block.  The full state is
\begin{equation}
  \rho_{\mathrm{loc}}=\begin{pmatrix}
    a'' & 0 & 0 & (1-p)^2 c \\[2pt]
    0 & \beta & 0 & 0 \\[2pt]
    0 & 0 & \beta & 0 \\[2pt]
    (1-p)^2 c^* & 0 & 0 & b''
  \end{pmatrix}
\end{equation}
where $a''=(1-p/2)^2\cos^2\!\theta+(p/2)^2\sin^2\!\theta$,
$b''=(p/2)^2\cos^2\!\theta+(1-p/2)^2\sin^2\!\theta$, and
$\beta=p(2-p)/4$ (see Appendix~\ref{app:calc}).

\subsection{Partial Transpose and Exact Negativity}

The $\{\ket{01},\ket{10}\}$ block of $\rho_{\mathrm{loc}}^{T_B}$ has eigenvalues
$\beta\pm(1-p)^2|c|$.  The unique negative eigenvalue gives
\begin{equation}
  \mathcal{N}_{\mathrm{loc}}(\theta,p)
  =\max\!\left[0,\;\frac{(1-p)^2}{2}\sin(2\theta)-\frac{p(2-p)}{4}\right].
  \label{eq:Nloc}
\end{equation}

\textbf{Physical interpretation.}
Compared with global depolarizing (Eq.~\eqref{eq:Ndep}), local depolarizing
suppresses the coherence more strongly ($(1-p)^2$ vs.\ $(1-p)$) and produces
a larger mixing term ($p(2-p)/4$ vs.\ $p/4$).
Consequently $\mathcal{N}_{\mathrm{loc}}\leq\mathcal{N}_{\mathrm{dep}}$ for all
$(\theta,p)$: local depolarizing is strictly more damaging than global depolarizing
per unit noise parameter.
Numerical validation (Appendix~\ref{app:numerics}) confirms $\Delta_{\max}<10^{-12}$;
Fig.~\ref{fig:landscape}(d) shows the full $\mathcal{N}_{\mathrm{loc}}(\theta,p)$
landscape, and Fig.~\ref{fig:three_channel} the $\theta=\pi/4$ cut alongside the
analytic curves.

\subsection{ESD Threshold}

Unlike phase damping (no finite sudden death) but like global depolarizing,
local depolarizing causes entanglement sudden death for \emph{all}
$\theta\in(0,\pi/2)$.  Setting $\lambda_-=0$:
\begin{equation}
  p_c^{\mathrm{loc}}(\theta)=1-\frac{1}{\sqrt{2\sin(2\theta)+1}}.
  \label{eq:pc_loc}
\end{equation}
Writing $p_c^{\mathrm{dep}}=2s/(2s+1)=1-1/(2s+1)$ with $s=\sin(2\theta)\in(0,1]$,
the two thresholds differ only in whether the subtracted term is $1/(2s+1)$ or
$1/\!\sqrt{2s+1}$.  Since $\sqrt{2s+1}<2s+1$, we have
$1/\!\sqrt{2s+1}>1/(2s+1)$, hence
$p_c^{\mathrm{loc}}(\theta)<p_c^{\mathrm{dep}}(\theta)$:
local depolarizing extinguishes entanglement at a \emph{lower} noise strength
than global depolarizing for every $\theta$.

\subsection{Entanglement-Survival Fraction}

The entanglement-survival fraction under local depolarizing is
\begin{equation}
  \eta_{\mathrm{loc}}(\theta,p)
  =\max\!\left[0,\;(1-p)^2-\frac{p(2-p)}{2\sin(2\theta)}\right].
  \label{eq:eta_loc}
\end{equation}
At $\theta=\pi/4$: $\eta_{\mathrm{loc}}(\pi/4,p)=\max[0,(1-p)^2-p(2-p)/2]
=\max[0,1-3p+3p^2/2]$.

\section{Entanglement Sudden Death}
\label{sec:ESD}

\subsection{Unified Structure of the Four Channels}

Although the four channels were derived separately, their negativities share a
single algebraic form.  Each post-gate channel acts on the restricted $X$-state in
two competing ways: it multiplies the entanglement-carrying anti-diagonal coherence
$\tfrac{1}{2}\sin2\theta$ by a \emph{coherence-suppression factor} $f$, and it
shifts the relevant partial-transpose block diagonal upward by a \emph{spectrum-shift
term} $g$.  The surviving negativity is the residual,
\begin{equation}
  \mathcal{N}(\theta,x)=\max\!\left[0,\;
  f(x)\,\tfrac{1}{2}\sin2\theta-g(\theta,x)\right],
  \label{eq:master}
\end{equation}
where $x$ denotes the channel's noise parameter ($p$ or $\gamma$).  The four
channels differ only in the pair $(f,g)$, collected in Table~\ref{tab:fg}.
Entanglement sudden death occurs precisely when the shift overtakes the suppressed
coherence, $g(\theta,x)=f(x)\,\tfrac{1}{2}\sin2\theta$; phase damping has $g\equiv0$
and therefore never reaches it.

\begin{table}[t]
\caption{The four post-gate channels reduced to the master
form~\eqref{eq:master}: coherence-suppression factor $f$ and spectrum-shift term
$g$.  Phase damping has no shift ($g=0$) and hence no sudden death; amplitude
damping is the only channel whose shift is input-angle dependent, which is why its
sudden-death region is confined to $\theta>\pi/4$.}
\label{tab:fg}
\begin{ruledtabular}
\begin{tabular}{lccc}
Channel & Parameter & $f$ & $g(\theta)$ \\
\colrule
Phase damping        & $p$      & $1-p$     & $0$ \\
Global depolarizing  & $p$      & $1-p$     & $p/4$ \\
Amplitude damping    & $\gamma$ & $1-\gamma$& $\gamma(1-\gamma)\sin^2\!\theta$ \\
Local depolarizing   & $p$      & $(1-p)^2$ & $p(2-p)/4$ \\
\end{tabular}
\end{ruledtabular}
\end{table}

This unified view makes the fragility ordering transparent.  Stronger suppression
(smaller $f$) and a larger shift (bigger $g$) both reduce the residual: local
depolarizing combines the smallest $f$, namely $(1-p)^2$, with the largest constant
shift $p(2-p)/4$, so it dies first; phase damping has the largest $f$ and zero
shift, so it never dies.  Amplitude damping is the outlier whose shift carries an
explicit $\sin^2\!\theta$, concentrating its damage on the more excited inputs and
producing the reversed ordering discussed below.

The structure is displayed geometrically in Fig.~\ref{fig:fgplane}, which plots each
channel as a trajectory in the $(f,g)$ plane.  The picture exposes a relation not
evident from Table~\ref{tab:fg}: global and local depolarizing trace the
\emph{same} line $g=(1-f)/4$, since $p(2-p)=1-(1-p)^2$, and differ only in how
quickly they move along it as the noise grows.  Phase damping runs along $g=0$ and
never enters the separable region, while amplitude damping is the only channel whose
trajectory is curved.  Sudden death corresponds to a trajectory crossing the line
$g=\tfrac{1}{2}f\sin2\theta$; the two depolarizing channels and amplitude damping
cross it (the latter only for $\theta>\pi/4$, where its $\sin^2\!\theta$ shift is
large enough), whereas phase damping never does.  The remaining subsections derive
the corresponding thresholds in closed form.

\subsection{Depolarizing: Critical Threshold}

Setting $\lambda_-=0$ in the depolarizing partial-transpose spectrum
(Appendix~\ref{app:calc}) gives
\begin{equation}
  p_c(\theta)=\frac{2\sin(2\theta)}{2\sin(2\theta)+1}.
  \label{eq:pc}
\end{equation}
Eq.~\eqref{eq:pc} is the depolarizing ESD threshold (Table~\ref{tab:results_summary}).

\begin{theorem}[Depolarizing ESD]
\label{thm:ESD_dep}
For every $\theta\in(0,\pi/2)$, $p_c(\theta)\in(0,1)$ and
$\Ndep(\theta,p)>0\Leftrightarrow p<p_c(\theta)$.
\end{theorem}
\begin{proof}
$\sin(2\theta)>0$ for $\theta\in(0,\pi/2)$ implies $p_c\in(0,1)$.  The
result follows from Eq.~\eqref{eq:Ndep}.
\end{proof}

The threshold $p_c$ is monotonically increasing in $\sin(2\theta)$: more
coherent inputs survive to larger depolarizing strengths.  The maximum
$p_c(\pi/4)=2/3$ corresponds to the known separability boundary of the
depolarized Bell-state mixture.

\subsection{Amplitude Damping: Critical Threshold}

Setting $\lambda_-=0$ in Eq.~\eqref{eq:evals_AD}:
$\gamma_c\sin^2\!\theta=\tfrac{1}{2}\sin(2\theta)=\sin\theta\cos\theta$, giving
\begin{equation}
  \gamma_c(\theta)=\cot\theta.
  \label{eq:gc}
\end{equation}
This is physical ($\gamma_c<1$) only when $\cot\theta<1$, i.e., $\theta>\pi/4$.

\begin{theorem}[Amplitude-Damping ESD]
\label{thm:ESD_AD}
Entanglement sudden death under independent amplitude damping occurs if and only
if $\theta>\pi/4$, with threshold $\gamma_c(\theta)=\cot\theta$.
For $\theta\leq\pi/4$, the negativity is positive for all $\gamma\in[0,1)$.
\end{theorem}
\begin{proof}
$\lambda_-<0\Leftrightarrow\gamma<\cot\theta$.  For $\theta\leq\pi/4$,
$\cot\theta\geq1$, so $\gamma<\cot\theta$ for all $\gamma\in[0,1)$.  For
$\theta>\pi/4$, $\cot\theta\in(0,1)$ gives a finite threshold.
\end{proof}

\subsection{Comparative Summary}

\begin{table*}[t]
\caption{Entanglement sudden death across the four post-gate noise channels.
  All four channels are compared with respect to the existence of a
  finite-noise sudden-death threshold and the robustness ordering of inputs.
  The orderings for depolarizing and amplitude damping are reversed.}
\label{tab:ESD}
\begin{ruledtabular}
\begin{tabular}{llll}
Channel & Sudden death? & Threshold & Robustness ordering\\
\hline
Phase damping  & Never ($p<1$) & --- & --- \\[2pt]
Global depolarizing & All $\theta\in(0,\tfrac{\pi}{2})$ &
  $\displaystyle p_c(\theta)=\frac{2\sin(2\theta)}{2\sin(2\theta)+1}$ &
  More coherent ($\uparrow\sin2\theta$) $\Rightarrow$ more robust \\[8pt]
Local depolarizing & All $\theta\in(0,\tfrac{\pi}{2})$ &
  $\displaystyle p_c^{\mathrm{loc}}(\theta)=1-\frac{1}{\sqrt{2\sin(2\theta)+1}}$ &
  More coherent $\Rightarrow$ more robust ($p_c^{\mathrm{loc}}<p_c$) \\[8pt]
Amplitude damping & Only $\theta>\tfrac{\pi}{4}$ &
  $\gamma_c(\theta)=\cot\theta$ &
  More excited ($\uparrow\sin\theta$) $\Rightarrow$ more fragile \\
\end{tabular}
\end{ruledtabular}
\end{table*}

The reversed robustness ordering is the key qualitative finding, summarized for all
four channels in Table~\ref{tab:ESD}.  Under global
depolarizing, maximally coherent inputs ($\theta=\pi/4$) are most robust with
$p_c=2/3$.  Under amplitude damping, exactly these inputs are the boundary case
($\gamma_c=1$, no finite sudden death), while inputs with $\theta>\pi/4$ (more
excited-state weight) are \emph{more fragile}: a larger $\sin^2\!\theta$ means
stronger population transfer, pushing $\lambda_-$ upward faster.

\begin{figure}[t]
  \centering
\includegraphics[width=0.95\columnwidth]{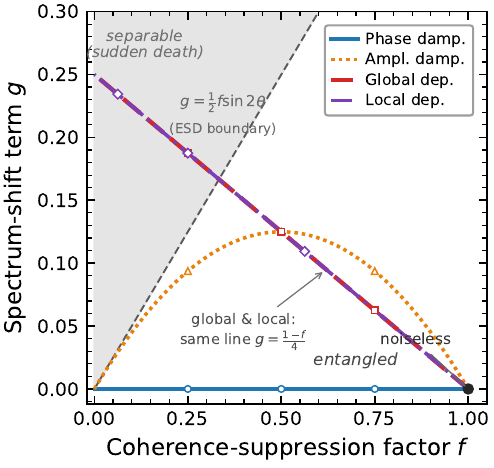}
  \caption{Geometry of the unified master relation~\eqref{eq:master} in the
    $(f,g)$ plane at $\theta=\pi/4$, where $f$ is the coherence-suppression factor
    and $g$ the spectrum-shift term (Table~\ref{tab:fg}).  Each channel traces a
    trajectory as its noise parameter runs from $0$ (the common noiseless point
    $(1,0)$, lower right) to $1$; open symbols mark $x=0.25,0.5,0.75$ and reveal the
    different spacing of noise along each path.  The dashed diagonal is the
    sudden-death boundary $g=\tfrac{1}{2}f\sin2\theta$: trajectories in the lower
    (white) region carry entanglement, while the shaded region is separable.  Phase
    damping ($g\equiv0$) runs along the horizontal axis and never crosses into the
    separable region, so it exhibits no sudden death.  Amplitude damping is the only
    channel with a curved (arched) trajectory.  Global and local depolarizing
    \emph{share the same straight trajectory} $g=(1-f)/4$---a non-obvious
    consequence of $p(2-p)=1-(1-p)^2$---and differ only in how rapidly they advance
    along it, local depolarizing reaching the boundary at smaller noise.}
  \label{fig:fgplane}
\end{figure}

\section{Entanglement Survival: Four-Channel Comparison}
\label{sec:efficiency}

The entanglement-survival fraction is defined as
\begin{equation}
  \eta(\theta,p)=\max\!\left[0,\,\frac{\mathcal{N}(\theta,p)}{\mathcal{N}_0(\theta)}\right],
  \quad\theta\in(0,\pi/2),
  \label{eq:eta_def}
\end{equation}
where $\theta\in(0,\pi/2)$ is required for $\mathcal{N}_0>0$.  At the endpoints
$\theta=0,\pi/2$ the input is incoherent and $\eta$ is undefined.  The explicit
$\max[0,\cdot]$ ensures $\eta\in[0,1]$: once entanglement sudden death occurs the survival fraction is zero, not negative.  We emphasize that $\eta$ measures the
\emph{fraction of the noiseless entanglement that survives the channel}, normalized
to the ideal output $\mathcal{N}_0$ at the same input angle; it is a normalized
entanglement-survival ratio rather than a rate of coherence consumption, since the
input coherence is fully expended by the ideal CNOT before the noise acts.

Substituting the four negativity expressions:
\begin{align}
  \eta_{\mathrm{phase}}(p)&=1-p,\label{eq:eta_phase}\\[4pt]
  \eta_{\mathrm{dep}}(\theta,p)&=\max\!\left[0,\,(1-p)-\frac{p}{2\sin(2\theta)}\right],
    \label{eq:eta_dep}\\[4pt]
  \eta_{\mathrm{AD}}(\theta,\gamma)&=\max\!\left[0,\,(1-\gamma)(1-\gamma\tan\theta)\right],
    \label{eq:eta_AD}\\[4pt]
  \eta_{\mathrm{loc}}(\theta,p)&=\max\!\left[0,\,(1-p)^2-\frac{p(2-p)}{2\sin(2\theta)}\right].
    \label{eq:eta_loc_xi}
\end{align}

\textbf{Phase damping:} $\eta_{\mathrm{phase}}=1-p$ is universal and
\emph{independent of $\theta$}.  Dephasing reduces the survival fraction
uniformly regardless of input coherence.

\textbf{Global depolarizing:} $\eta_{\mathrm{dep}}$ depends on $\theta$ through
$\sin(2\theta)$.  More coherent inputs retain a larger fraction of the ideal
negativity for fixed $p$.

\textbf{Amplitude damping:} $\eta_{\mathrm{AD}}$ depends on $\theta$ through
$\tan\theta$.  For $\theta<\pi/4$ ($\tan\theta<1$), the factor $(1-\gamma\tan\theta)>0$
for all $\gamma\in[0,1)$: no sudden death.  For $\theta>\pi/4$ ($\tan\theta>1$),
the population-transfer penalty drives $\eta_{\mathrm{AD}}$ to zero at
$\gamma_c=\cot\theta$.

\textbf{Local depolarizing:} $\eta_{\mathrm{loc}}$ has the strongest degradation,
combining the quadratic coherence factor $(1-p)^2$ with the enhanced mixing
penalty $p(2-p)/[2\sin(2\theta)]$.  It vanishes at the lowest threshold of all
four channels, $p_c^{\mathrm{loc}}(\theta)<p_c(\theta)$ for every $\theta$.

\textbf{Direct comparison at $\theta=\pi/4$.}
Plotting negativity and survival fraction directly against the physical noise
parameter (Fig.~\ref{fig:three_channel}) makes the distinct degradation behaviour
of each channel immediately visible.  At $\theta=\pi/4$ the four channels
order by fragility: local depolarizing decays fastest and reaches sudden
death first (at $p_c^{\mathrm{loc}}=1-1/\sqrt{3}\approx0.42$), followed by
global depolarizing (at $p_c=2/3$), while phase damping and amplitude damping
retain entanglement up to the maximal noise value.  The survival-fraction curves
$\eta(p)$ quantify how much of the noiselessly-converted entanglement each mechanism
preserves as the noise parameter increases:
\begin{itemize}
  \item Phase damping degrades \emph{linearly}, $\eta_{\mathrm{phase}}=1-p$,
    with no sudden death.
  \item Global depolarizing degrades linearly until the isotropic admixture
    forces sudden death at $p_c=2/3$.
  \item Amplitude damping degrades with initial slope $-2$ at $\theta=\pi/4$
    ($\eta_{\mathrm{AD}}=(1-\gamma)^2=1-2\gamma+\gamma^2$), steeper than phase and
    global damping per unit noise; the $\gamma^2$ term is the only second-order
    contribution.
  \item Local depolarizing degrades fastest, reflecting the quadratic
    coherence suppression $(1-p)^2$ combined with the larger mixing term.
\end{itemize}

\textbf{Comparison at $\theta=\pi/4$:}
\begin{align}
  \eta_{\mathrm{phase}}(p)&=1-p,\label{eq:compare_ph}\\[2pt]
  \eta_{\mathrm{dep}}\!\left(\tfrac{\pi}{4},p\right)
    &=\max\!\left[0,1-\tfrac{3}{2}p\right],\label{eq:compare_dep}\\[2pt]
  \eta_{\mathrm{AD}}\!\left(\tfrac{\pi}{4},\gamma\right)&=(1-\gamma)^2.\label{eq:compare_AD}
\end{align}
For small noise: $\eta_{\mathrm{phase}}\approx1-p$,
$\eta_{\mathrm{dep}}\approx1-\tfrac{3}{2}p$,
$\eta_{\mathrm{AD}}\approx1-2\gamma$.  Amplitude damping is twice as harmful as
phase damping per unit noise parameter; global depolarizing is $1.5\times$ as harmful
at $\theta=\pi/4$.

\textbf{Initial-slope signature.}
The four channels have \emph{distinct} initial survival-fraction slopes at $\theta=\pi/4$,
$\mathrm{d}\eta/\mathrm{d}x|_{x=0}=-1,\,-\tfrac{3}{2},\,-2,\,-3$ for phase, global,
amplitude, and local respectively (with $\eta_{\mathrm{loc}}=1-3p+\tfrac{3}{2}p^2$).
Because these four values are all different, the dominant post-gate noise mechanism
can in principle be identified from the leading-order decay of the conversion
survival fraction measured at small noise---a simple diagnostic that requires only the
slope of $\eta(x)$ near $x=0$ rather than the full curve.

\textbf{A noise-strength-dependent robustness crossing.}
The ordering by robustness is not fixed across the full noise range.  At
$\theta=\pi/4$ phase damping is always the most robust and local depolarizing
always the most fragile, but the \emph{relative} ordering of global depolarizing and
amplitude damping depends on noise strength.  Because the two channels are
controlled by different physical parameters---the depolarizing admixture $p$ and the
relaxation probability $\gamma$---this comparison is only well defined once a common
scale is chosen.  Comparing them at equal numerical noise strength $x=p=\gamma$, the
survival fractions of Eqs.~\eqref{eq:compare_dep}--\eqref{eq:compare_AD} satisfy
$1-\tfrac{3}{2}x=(1-x)^2$ at
\begin{equation}
  x_\times=\tfrac{1}{2}.
  \label{eq:crossing}
\end{equation}
For $0<x<\tfrac{1}{2}$ one has $1-\tfrac{3}{2}x>(1-x)^2$, so global depolarizing
retains more entanglement; for $\tfrac{1}{2}<x<\tfrac{2}{3}$ the inequality reverses,
$(1-x)^2>1-\tfrac{3}{2}x$; and for $x\ge\tfrac{2}{3}$ global depolarizing has already
undergone sudden death ($\eta_{\mathrm{dep}}=0$) while amplitude damping remains
positive for $x<1$.  Both survival fractions have linear leading behaviour near the origin
($-\tfrac{3}{2}x$ and $-2x$ respectively, so amplitude damping falls faster
initially); the crossing reflects the milder global-depolarizing slope being
overtaken once that channel approaches its threshold, not any quadratic onset.  We
stress that $x_\times=\tfrac{1}{2}$ is a same-numerical-parameter statement only: a
hardware comparison requires mapping $p$ and $\gamma$ to common device quantities
(e.g.\ $F_{\mathrm{avg}}$, $T_1$, $T_\varphi$, or an elapsed time $t$), under which
the crossing location shifts and may disappear depending on the device timescales.

\section{Connection to Physical Parameters}
\label{sec:physical_params}

The phenomenological noise parameters connect to physical device quantities as
follows; the full derivations of each mapping are given in
Appendix~\ref{app:device}.

\textbf{Phase damping.}
For noise duration $t$ after the gate,
$p=1-e^{-2t/T_\varphi}$ where $T_\varphi$ is the pure dephasing time
defined by $1/T_2=1/(2T_1)+1/T_\varphi$; the factor of two in the exponent
follows from the $\sqrt{1-p}$ coherence suppression of the Kraus parametrization
(Appendix~\ref{app:device}).

\textbf{Amplitude damping.}
$\gamma=1-e^{-t/T_1}$ where $T_1$ is the energy relaxation time.

\textbf{Global depolarizing.}
For the channel $\mathcal{E}_{\mathrm{dep}}(\rho)=(1-p)\rho+p\,\mathbb{I}_4/4$ used
in our analytic formulas, the average gate fidelity is
$F_{\mathrm{avg}}=1-\tfrac{3}{4}p$ (Appendix~\ref{app:device}), so the depolarizing
parameter that enters the negativity is
\begin{equation}
  p=\frac{4}{3}\bigl(1-F_{\mathrm{avg}}\bigr)=\frac{4}{3}\,\varepsilon_{\mathrm{gate}},
  \label{eq:p_gate}
\end{equation}
where $\varepsilon_{\mathrm{gate}}=1-F_{\mathrm{avg}}$ is the average gate
infidelity.  We caution that a different and equally common convention reports the
depolarizing strength relative to the $d^2-1$ non-identity Pauli channels of
randomised benchmarking, for which $p_{\mathrm{RB}}=\tfrac{16}{15}(1-F_{\mathrm{avg}})$;
this $p_{\mathrm{RB}}$ is a distinct quantity and should not be substituted into the
formulas above.  A gate error rate $\varepsilon_{\mathrm{gate}}=0.5\%$ then gives
$p\approx6.7\times10^{-3}$.

\textbf{Numerical estimates for superconducting qubits.}
Typical values $T_1\approx100\,\mu\mathrm{s}$, $T_\varphi\approx80\,\mu\mathrm{s}$,
and post-gate noise duration $t\approx200\,\mathrm{ns}$ give
$\gamma\approx2\times10^{-3}$ and $p_{\mathrm{deph}}\approx5\times10^{-3}$.
These place both channels in the regime $p,\gamma\ll1$ where the expressions
are well approximated by their leading-order forms.
On other hardware platforms---for example photonic systems---the effective
noise parameters are set by the relevant loss and dephasing rates; our formulas
apply directly once those are mapped to the Kraus operators
in Eqs.~\eqref{eq:PD_Kraus}--\eqref{eq:AD_Kraus} (collected for all four channels
in Appendix~\ref{app:kraus}).
Practical noise suppression through environmental shielding or dynamical
decoupling corresponds directly to a reduction of the phenomenological
parameters $p$ and $\gamma$.

All analytic expressions in this paper are \emph{exact} (non-perturbative) in $p$
and $\gamma$, so they remain valid at arbitrary noise strengths.

\textbf{Limitations of the model.}
The present analysis assumes: (i)~an ideal CNOT gate; (ii)~two-qubit systems only;
(iii)~post-gate Markovian noise with no gate--noise temporal correlations;
(iv)~four one-parameter channels; (v)~no measurement error, readout noise,
leakage, or cross-talk.  These idealisations are explicit and define the scope of
the benchmark.  Real devices involve combinations of these effects, and the
present expressions serve as reference points rather than direct device predictions.

\section{Summary of Exact Analytic Results}
\label{sec:results_summary}

Table~\ref{tab:results_summary} collects the closed-form negativity, survival fraction, and sudden-death threshold for each of the four channels
in a unified format organised by channel.  The noise parameters satisfy
$p,\gamma\in[0,1]$ and the input coherence angle $\theta\in[0,\pi/2]$.
The survival fraction $\eta=\mathcal{N}/\mathcal{N}_0$ is defined for
$\theta\in(0,\pi/2)$; the amplitude-damping ESD threshold $\gamma_c=\cot\theta$
is physical ($<1$) only for $\theta>\pi/4$.

\begin{table*}[t]
\caption{%
  Closed-form analytic results for coherence-to-entanglement conversion in the
  CNOT protocol under the four post-gate noise channels, with the ideal baseline
  $\mathcal{N}_0(\theta)=\tfrac{1}{2}\sin(2\theta)=\tfrac{1}{2}C_{\ell_1}$.
  $p,\gamma\in[0,1]$, $\theta\in[0,\pi/2]$; survival fraction
  $\eta=\mathcal{N}/\mathcal{N}_0$ requires $\theta\in(0,\pi/2)$.  All negativities
  and survival fractions are clamped at zero; the AD threshold is physical only for
  $\theta>\pi/4$.%
}
\label{tab:results_summary}
\begin{ruledtabular}
\begin{tabular}{lccc}
Channel & Negativity $\mathcal{N}(\theta,x)$ & Survival fraction $\eta(\theta,x)$
  & ESD threshold \\[2pt]
\hline \\[-6pt]
Phase damping &
  $\dfrac{1-p}{2}\sin2\theta$ &
  $1-p$ &
  none \\[10pt]
Global depolarizing &
  $\dfrac{1-p}{2}\sin2\theta-\dfrac{p}{4}$ &
  $(1-p)-\dfrac{p}{2\sin2\theta}$ &
  $p_c=\dfrac{2\sin2\theta}{2\sin2\theta+1}$ \\[12pt]
Local depolarizing &
  $\dfrac{(1-p)^2}{2}\sin2\theta-\dfrac{p(2-p)}{4}$ &
  $(1-p)^2-\dfrac{p(2-p)}{2\sin2\theta}$ &
  $p_c^{\mathrm{loc}}=1-\dfrac{1}{\sqrt{2\sin2\theta+1}}$ \\[12pt]
Amplitude damping &
  $(1-\gamma)\!\left(\dfrac{\sin2\theta}{2}-\gamma\sin^2\!\theta\right)$ &
  $(1-\gamma)(1-\gamma\tan\theta)$ &
  $\gamma_c=\cot\theta$ \\[6pt]
\end{tabular}
\end{ruledtabular}
\end{table*}

\section{Discussion}
\label{sec:discussion}

The results provide a compact four-channel analytic picture of how coherence-to-entanglement
conversion is affected by distinct post-gate noise mechanisms.

\textbf{Spectral mechanism.}
All four channels act on the same $X$-state partial-transpose structure through the
single competition of Eq.~\eqref{eq:master}: a coherence-suppression factor $f$ on
the entanglement-carrying off-diagonal against an additive spectrum shift $g$
(Table~\ref{tab:fg}).  The qualitative differences follow directly from the pair
$(f,g)$: phase damping has $g=0$ and never reaches sudden death; global depolarizing
adds a constant shift $p/4$ and dies at a finite threshold for every input; local
depolarizing combines the stronger quadratic suppression $(1-p)^2$ with the same
shift, giving the strictly lower threshold
$p_c^{\mathrm{loc}}(\theta)<p_c^{\mathrm{dep}}(\theta)$ and sudden death for
\emph{all} $\theta$; and amplitude damping is the sole channel whose shift carries
an explicit $\sin^2\!\theta$, confining its sudden death to $\theta>\pi/4$.

\textbf{Reversed robustness ordering and ESD hierarchy.}
Among the two depolarizing channels, which share the same parameter $p$, the
comparison is direct: maximally coherent inputs ($\theta=\pi/4$) are most robust
under global depolarizing ($p_c^{\mathrm{dep}}=2/3$), and local depolarizing
produces a strictly lower threshold $p_c^{\mathrm{loc}}(\theta)<p_c^{\mathrm{dep}}(\theta)$
at every~$\theta$, because the quadratic off-diagonal suppression $(1-p)^2$ makes
it more damaging per unit $p$.
Amplitude damping is parametrized by a \emph{different} physical quantity $\gamma$,
so its threshold $\gamma_c=\cot\theta$ cannot be placed on the same numerical scale
as the depolarizing $p_c$ values; the contrast is qualitative rather than a direct
inequality.  Qualitatively, amplitude damping reverses the angular trend: maximally
coherent inputs are the boundary case ($\gamma_c=1$ at $\theta=\pi/4$, no finite
sudden death), while more excited inputs ($\theta>\pi/4$) become progressively
fragile---the \emph{opposite} of the depolarizing channels, where coherent inputs
are the most robust.  This sign reversal in the angular dependence, rather than any
cross-parameter inequality, is the essential distinction between the
population-conserving (depolarizing) and population-transferring (amplitude-damping)
mechanisms.

To compare all four channels on a common physical footing, one maps each parameter
to a single elapsed time $t$ via the device relations of
Section~\ref{sec:physical_params} ($p=1-e^{-2t/T_\varphi}$, $\gamma=1-e^{-t/T_1}$).
At fixed $t$ and $\theta=\pi/4$, ordering the channels by the noise duration at
which entanglement vanishes reproduces the fragility sequence visible in
Fig.~\ref{fig:three_channel}: local depolarizing dies first, then global
depolarizing, while phase damping and (at this angle) amplitude damping retain
entanglement until the noise parameter saturates.

\textbf{Degradation profiles.}
Plotted against the physical noise parameter (Fig.~\ref{fig:three_channel}),
the four channels show clearly distinct decay shapes at $\theta=\pi/4$.
Phase damping and global depolarizing both begin with a linear decay in $p$;
phase damping continues linearly to $p=1$ with no sudden death, while global
depolarizing is driven to zero at $p_c=2/3$ by the isotropic identity admixture.
Amplitude damping has the steeper initial survival-fraction slope $-2$ at $\theta=\pi/4$,
$\eta_{\mathrm{AD}}=(1-\gamma)^2$; the population-transfer correction contributes
already at first order in $\gamma$, with only the $\gamma^2$ self-product entering
at second order.
Local depolarizing decays fastest of all, owing to the quadratic coherence
suppression $(1-p)^2$ together with the enhanced mixing term, reaching sudden
death at the lowest noise value $p_c^{\mathrm{loc}}\approx0.42$.

\textbf{Positioning relative to prior work.}
Streltsov et al.~\cite{streltsov2015measuring} established the general principle
that coherence can be converted to entanglement via incoherent operations.
Prior work on coherence consumption in quantum
algorithms~\cite{naseri2022entanglement,shi2017coherence,feng2023coherence,
liu2019coherence,ahnefeld2022coherence,ye2026coherence} studies coherence as
a computational resource rather than as an input to be converted into entanglement.
The entanglement-sudden-death phenomenon~\cite{yu2006quantum} was identified in
a general context; the threshold formula $p_c(\theta)$ derived here is a
coherence-input-dependent specialisation not previously given for this conversion
protocol.

Noise suppression strategies---such as environmental shielding or dynamical
decoupling---translate directly into reductions of the phenomenological
parameters $p$ and $\gamma$ in the expressions above.

\textbf{Scope of the unified structure.}
The master form~\eqref{eq:master} applies to post-gate channels that preserve the
restricted block structure of the CNOT output---that is, channels under which the
entanglement remains confined to the $\{\ket{00},\ket{11}\}$ sector while the
$\{\ket{01},\ket{10}\}$ sector acquires only incoherent populations.  The four
channels studied here all satisfy this condition.  Channels that generate coherence
in the $\{\ket{01},\ket{10}\}$ block---for instance independent bit-flip noise,
which populates $\ket{01}\!\bra{10}$---fall outside the single-pair $(f,g)$
description and would require the full four-by-four partial-transpose spectrum.  The
unified picture is therefore a statement about this physically natural family of
population-preserving and dephasing channels, not about arbitrary Kraus maps.

\textbf{An open question: optimality of the conversion ratio.}
The noiseless identity $\mathcal{N}_0=\tfrac{1}{2}C_{\ell_1}$ fixes the
coherence-to-entanglement ratio for the CNOT protocol at exactly $\tfrac{1}{2}$.
Whether this ratio is \emph{optimal}---that is, whether $\tfrac{1}{2}$ is the
largest output negativity obtainable per unit input coherence under any incoherent
two-qubit operation, or whether a different entangling gate or measurement-assisted
protocol could exceed it---is not addressed here and remains open.  Establishing
$\tfrac{1}{2}$ as a tight bound, or exhibiting a protocol that surpasses it, would
turn the present baseline identity into an operational limit on coherence-to-
entanglement conversion.  We regard this as the most natural theoretical extension
of the present results.

\section{Conclusion}
\label{sec:conclusion}

We have derived closed-form analytic benchmarks for coherence-to-entanglement
entanglement-survival fraction in a minimal two-qubit CNOT protocol under four post-gate
noise channels.  The main findings are:
\begin{enumerate}
  \item \textbf{Noiseless relation.} The CNOT converts input $\ell_1$-norm coherence to
    output negativity with $\mathcal{N}_0=\tfrac{1}{2}C_{\ell_1}$.
  \item \textbf{Phase damping} produces a universal, coherence-independent survival fraction
    $\eta_{\mathrm{phase}}=1-p$ and no entanglement sudden death.
  \item \textbf{Global depolarizing} produces coherence-dependent survival fraction and
    entanglement sudden death at $p_c(\theta)=2\sin(2\theta)/[2\sin(2\theta)+1]$
    for all $\theta\in(0,\pi/2)$.  More coherent inputs are more robust.
  \item \textbf{Amplitude damping} produces a population-transfer correction and
    sudden death only for $\theta>\pi/4$ at $\gamma_c=\cot\theta$.  More excited
    inputs are more fragile---the \emph{opposite} robustness ordering from
    depolarizing.
  \item \textbf{Local single-qubit depolarizing} yields exact negativity
    $\mathcal{N}_{\mathrm{loc}}=\max[0,(1-p)^2\sin(2\theta)/2-p(2-p)/4]$
    and ESD threshold $p_c^{\mathrm{loc}}(\theta)=1-1/\!\sqrt{2\sin(2\theta)+1}$
    for \emph{all} $\theta\in(0,\pi/2)$, with
    $p_c^{\mathrm{loc}}<p_c^{\mathrm{dep}}$ at every $\theta$.
  \item \textbf{Fragility ordering and a robustness crossing.}  At $\theta=\pi/4$
    phase damping is always the most robust and local depolarizing the most fragile
    (sudden death at $p\approx0.42$, versus $p_c=2/3$ for global depolarizing).  The
    relative ordering of global depolarizing and amplitude damping is not fixed:
    compared at equal numerical noise strength $x=p=\gamma$, their survival fractions cross
    at $x=1/2$, with global depolarizing more robust below this value and amplitude
    damping above it.  Since $p$ and $\gamma$ are distinct physical parameters, a
    hardware comparison requires mapping both to common device quantities, under
    which the crossing location shifts.  The four distinct initial survival-fraction slopes
    $(-1,-\tfrac{3}{2},-2,-3)$ further allow the dominant noise mechanism to be
    identified from the small-noise decay alone.
  \item \textbf{Physical mapping.} The parameters satisfy
    $p=1-e^{-2t/T_\varphi}$, $\gamma=1-e^{-t/T_1}$, and
    $p=\tfrac{4}{3}(1-F_{\mathrm{avg}})$, connecting the benchmarks to measurable
    device quantities such as for superconducting qubit.
\end{enumerate}

\textbf{Limitations.}
The present analysis assumes an ideal CNOT gate; a two-qubit system only; post-gate
Markovian noise; four single-parameter channels; and no measurement error, readout
noise, leakage, cross-talk, or control errors.  These idealisations are explicit;
the closed-form expressions serve as analytic reference points rather than
hardware predictions.

\textbf{Future directions.}
Natural extensions include: noisy CNOT gates;
correlated and non-Markovian dephasing; amplitude-phase mixed channels;
multi-qubit CNOT networks; non-Markovian dynamics; hardware-calibrated noise
models using the parameter mappings of Section~\ref{sec:physical_params};
and protocols beyond the minimal two-qubit setting.

\appendix

\section{Supporting Calculations}
\label{app:calc}

\subsection{Noiseless CNOT Output}

$U_{\mathrm{CNOT}}\ket{a,b}=\ket{a,a\oplus b}$.  Applying to
$\ket{\psi_A}\otimes\ket{0}$ gives $\ket{\Psi}$ as in Eq.~\eqref{eq:Psi}
with density matrix having entries $a=\cos^2\theta$, $b=\sin^2\theta$,
$c=\tfrac{1}{2}e^{-i\phi}\sin(2\theta)$.  Proposition~\ref{prop:CE} then follows
from $\mathcal{N}(\rho_X)=|c|=\tfrac{1}{2}\sin(2\theta)=\tfrac{1}{2}C_{\ell_1}$.

\subsection{Proof of Proposition~\ref{prop:Xstate} and Phase-Damping Spectrum}

For the restricted $X$-state~\eqref{eq:rho_Xstate}, the partial transpose
$\rho_X^{T_B}$ is block-diagonal with a $2\times2$ block
$\bigl(\begin{smallmatrix}0&c'\\c'^*&0\end{smallmatrix}\bigr)$
and $1\times1$ blocks $a'$, $b'$.  The $2\times2$ block has eigenvalues $\pm|c'|$;
the unique negative eigenvalue gives $\mathcal{N}=|c'|$.  For phase damping,
$c'=(1-p)c$ so $\Nph=\tfrac{1}{2}(1-p)\sin(2\theta)$.

\subsection{Equivalence of Negativity and Concurrence ($C=2\mathcal{N}$)}
\label{app:CeqN}

Every state produced in this work---the noiseless output~\eqref{eq:rho_Xstate} and
the four noisy outputs---has the general $X$-state form
$\rho=\mathrm{diag}(a,f,f,b)+c\ket{00}\!\bra{11}+c^*\ket{11}\!\bra{00}$,
in which the only off-diagonal coherence lies in the $\{\ket{00},\ket{11}\}$ block
and the $\{\ket{01},\ket{10}\}$ block carries equal incoherent populations $f$
(with $f=0$ for phase damping and $f=p/4$, $\gamma(1-\gamma)\sin^2\!\theta$,
$p(2-p)/4$ for global depolarizing, amplitude damping, and local depolarizing,
respectively).  For such an $X$-state the Wootters
concurrence~\cite{wootters1998entanglement} evaluates to
\begin{equation}
  C=2\max\!\bigl\{|c|-f,\;-\sqrt{ab}\bigr\}=2\max\bigl(|c|-f,\,0\bigr),
  \label{eq:concurrence}
\end{equation}
the second alternative being non-positive.  The negativity of the same state is
$\mathcal{N}=\max(|c|-f,0)$ from the $\{\ket{01},\ket{10}\}$ partial-transpose
block.  Hence
\begin{equation}
  C=2\mathcal{N}
  \label{eq:CeqN}
\end{equation}
holds \emph{exactly} for the entire family, independently of the channel and of the
noise strength.  Entanglement is therefore carried entirely by the
$\{\ket{00},\ket{11}\}$ coherence, and the two measures coincide up to the constant
factor of two.  Every sudden-death threshold and robustness ordering in this paper
is consequently measure-independent: it is identical whether quantified by
negativity or concurrence.

\subsection{Depolarizing Partial-Transpose Spectrum}

The depolarizing output has anti-diagonal coherence
$c'=\tfrac{1}{2}(1-p)\sin(2\theta)$ (the input coherence suppressed by $(1-p)$)
and diagonal entries shifted by $p/4$.  The $2\times2$ block in the
$\{\ket{01},\ket{10}\}$ subspace of the partial transpose is
$\bigl(\begin{smallmatrix}p/4&c'\\c'^*&p/4\end{smallmatrix}\bigr)$
with eigenvalues $p/4\pm\tfrac{1}{2}(1-p)\sin(2\theta)$, giving the negative
eigenvalue $\lambda_-=p/4-\tfrac{1}{2}(1-p)\sin(2\theta)$ and
$\Ndep=\max[0,-\lambda_-]$.

\subsection{Amplitude Damping: Derivation of $\rho_{\mathrm{AD}}$}

With $A_{ij}=A_i\otimes A_j$, the four contributions
$(A_{ij})\rho(A_{ij})^\dagger$ to
$\rho_{\mathrm{AD}}=\sum_{i,j}(A_{ij})\rho(A_{ij})^\dagger$ are (writing the input
CNOT-output entries as $\cos^2\!\theta$, $\sin^2\!\theta$, and
$c=\tfrac{1}{2}e^{-i\phi}\sin(2\theta)$):
\begin{align*}
  (A_{00}):\ &
    \cos^2\!\theta\,\ket{00}\!\bra{00}+(1-\gamma)c\,\ket{00}\!\bra{11}\nonumber\\
    &\quad+(1-\gamma)c^*\ket{11}\!\bra{00}
    +(1-\gamma)^2\sin^2\!\theta\,\ket{11}\!\bra{11},\\
  (A_{01}):\ &\ \gamma(1-\gamma)\sin^2\!\theta\,\ket{10}\!\bra{10},\\
  (A_{10}):\ &\ \gamma(1-\gamma)\sin^2\!\theta\,\ket{01}\!\bra{01},\\
  (A_{11}):\ &\ \gamma^2\sin^2\!\theta\,\ket{00}\!\bra{00}.
\end{align*}
Summing, the $\ket{00}\!\bra{00}$ entry becomes
$a=\cos^2\!\theta+\gamma^2\sin^2\!\theta$, the $\ket{11}\!\bra{11}$ entry
$b=(1-\gamma)^2\sin^2\!\theta$, and the two central populations
$f=\gamma(1-\gamma)\sin^2\!\theta$, reproducing Eq.~\eqref{eq:rho_AD}.
The partial transpose moves $c$ into the $\{\ket{01},\ket{10}\}$ block, giving
eigenvalues $f\pm|c|=(1-\gamma)[\gamma\sin^2\!\theta\pm\tfrac{1}{2}\sin(2\theta)]$
as in Eq.~\eqref{eq:evals_AD}.

\subsection{Local Depolarizing: Derivation of $\rho_{\mathrm{loc}}$}

Each single-qubit depolarizing map sends
$\ket{0}\!\bra{0}\to(1-\tfrac{p}{2})\ket{0}\!\bra{0}+\tfrac{p}{2}\ket{1}\!\bra{1}$,
$\ket{1}\!\bra{1}\to\tfrac{p}{2}\ket{0}\!\bra{0}+(1-\tfrac{p}{2})\ket{1}\!\bra{1}$,
and $\ket{0}\!\bra{1}\to(1-p)\ket{0}\!\bra{1}$.
Applying $\mathcal{E}_p\otimes\mathcal{E}_p$ to the CNOT
output~\eqref{eq:rho_Xstate}, the coherence acquires a factor $(1-p)^2$ on each
qubit's off-diagonal, giving anti-diagonal element $(1-p)^2c$.  The
$\ket{00}\!\bra{00}$ population receives $(1-\tfrac{p}{2})^2$ from $\cos^2\!\theta$
and $(\tfrac{p}{2})^2$ from $\sin^2\!\theta$, so
$a''=(1-\tfrac{p}{2})^2\cos^2\!\theta+(\tfrac{p}{2})^2\sin^2\!\theta$, and
symmetrically $b''=(\tfrac{p}{2})^2\cos^2\!\theta+(1-\tfrac{p}{2})^2\sin^2\!\theta$.
Each of the $\{\ket{01},\ket{10}\}$ populations collects
$\tfrac{p}{2}(1-\tfrac{p}{2})$ from both $\cos^2\!\theta$ and $\sin^2\!\theta$,
giving $\beta=\tfrac{p}{2}(1-\tfrac{p}{2})=p(2-p)/4$.  The partial-transpose
$\{\ket{01},\ket{10}\}$ block is
$\bigl(\begin{smallmatrix}\beta&(1-p)^2c\\(1-p)^2c^*&\beta\end{smallmatrix}\bigr)$
with negative eigenvalue $\lambda_-=\beta-(1-p)^2|c|$, reproducing
Eq.~\eqref{eq:Nloc}.  Setting $\lambda_-=0$ with $u=1-p$ and $s=\sin(2\theta)$
gives $u^2(2s+1)=1$, hence $p_c^{\mathrm{loc}}=1-1/\!\sqrt{2s+1}$
[Eq.~\eqref{eq:pc_loc}].

\section{Numerical Validation}
\label{app:numerics}

We verify all closed-form results (negativity, survival fraction, and threshold for each channel) by direct density-matrix simulation.
The simulated protocol is
\begin{equation}
  \rho_{\mathrm{out}}
  =\mathcal{E}_{\mathrm{noise}}\!\left(U_{\mathrm{CNOT}}\rho_0
    U_{\mathrm{CNOT}}^\dagger\right),
\end{equation}
where $\rho_0=\rho_A\otimes\ket{0}\!\bra{0}$.  Simulations are performed on a
uniform grid over $\theta\in[0.01,\pi/2]$ and $p,\gamma\in[0,1]$.
The maximum absolute deviation between numerical and analytic negativity is
\begin{equation}
  \Delta_{\max}=\max_{\theta,p}\left|\mathcal{N}_{\mathrm{num}}(\theta,p)
    -\mathcal{N}_{\mathrm{analytic}}(\theta,p)\right|<10^{-12}
  \label{eq:Delta_max}
\end{equation}
for all four channels, consistent with double-precision roundoff.

Figure~\ref{fig:CE_prop} confirms the noiseless relation
$\mathcal{N}_0=\tfrac{1}{2}C_{\ell_1}$ to machine precision.
Figure~\ref{fig:neg_vs_theta} shows the $\theta$-dependence of the negativity for
all four channels at three noise strengths.
Figure~\ref{fig:landscape} shows the four-channel two-dimensional negativity
landscapes $\mathcal{N}(\theta,\cdot)$: smooth contours with no
ESD for phase damping; coherence-dependent ESD boundary $p_c(\theta)$ for global depolarizing;
ESD boundary $\gamma_c=\cot\theta$ for $\theta>\pi/4$ for amplitude damping;
and the stricter boundary $p_c^{\mathrm{loc}}(\theta)$ for local depolarizing.
Figure~\ref{fig:three_channel} validates the consolidated four-channel negativity
and entanglement-survival fraction at $\theta=\pi/4$ against direct simulation; its lower
panels plot the absolute deviations $|\Delta\mathcal{N}|$ and $|\Delta\eta|$
directly, which remain at the $10^{-16}$ roundoff level across the full noise range
and thus display the bound~\eqref{eq:Delta_max} graphically.
Figure~\ref{fig:ESD} validates the sudden-death thresholds and angle-resolved survival fraction.
Figure~\ref{fig:fgplane} recasts the same four channels geometrically in the
$(f,g)$ plane of the master relation~\eqref{eq:master}.

\begin{figure}[t]
  \centering
  \includegraphics[width=0.82\columnwidth]{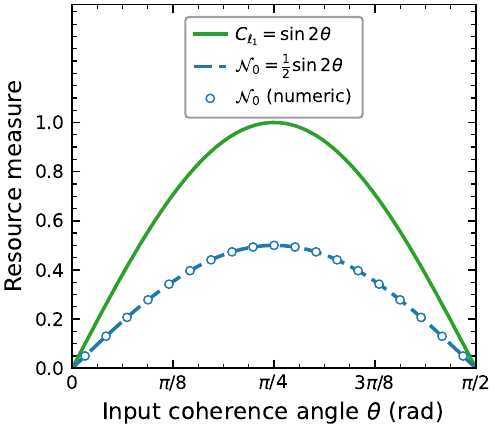}
  \caption{Ideal (noiseless) coherence-to-entanglement proportionality.
    The input $\ell_1$-norm coherence $C_{\ell_1}=\sin(2\theta)$ (green solid)
    and the post-gate negativity $\mathcal{N}_0=\tfrac{1}{2}\sin(2\theta)$
    (blue dashed) differ by exactly a factor of two.
    Open circles: numerically extracted $\mathcal{N}_0$, confirming
    $\mathcal{N}_0=\tfrac{1}{2}C_{\ell_1}$ to machine precision.
    This sets the baseline resource that the four noise channels degrade.}
  \label{fig:CE_prop}
\end{figure}

\begin{figure*}[t]
  \centering
  \includegraphics[width=0.96\textwidth]{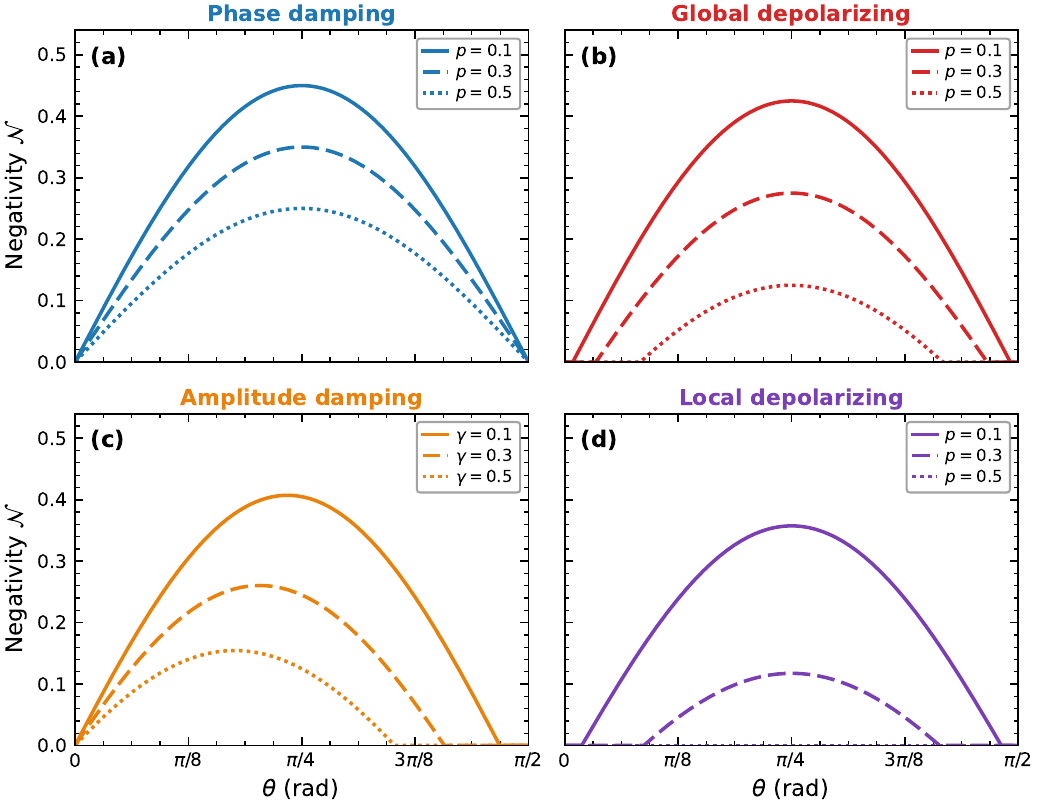}
  \caption{Negativity $\mathcal{N}(\theta,\cdot)$ versus input angle $\theta$
    for all four post-gate noise channels, one channel per panel and all panels
    on identical axes for direct comparison:
    (a) phase damping, (b) global depolarizing, (c) amplitude damping,
    (d) local single-qubit depolarizing.
    Within each panel the three curves correspond to noise level
    $0.1/0.3/0.5$ (solid/dashed/dotted); amplitude damping uses the decay
    probability $\gamma$, the other three use $p$.
    Phase damping simply rescales the $\sin(2\theta)$ envelope with no sudden
    death; global and local depolarizing develop a zero-negativity region that
    widens with noise; amplitude damping is fragile only for $\theta>\pi/4$.}
  \label{fig:neg_vs_theta}
\end{figure*}

\begin{figure*}[t]
  \centering
  \includegraphics[width=0.96\textwidth]{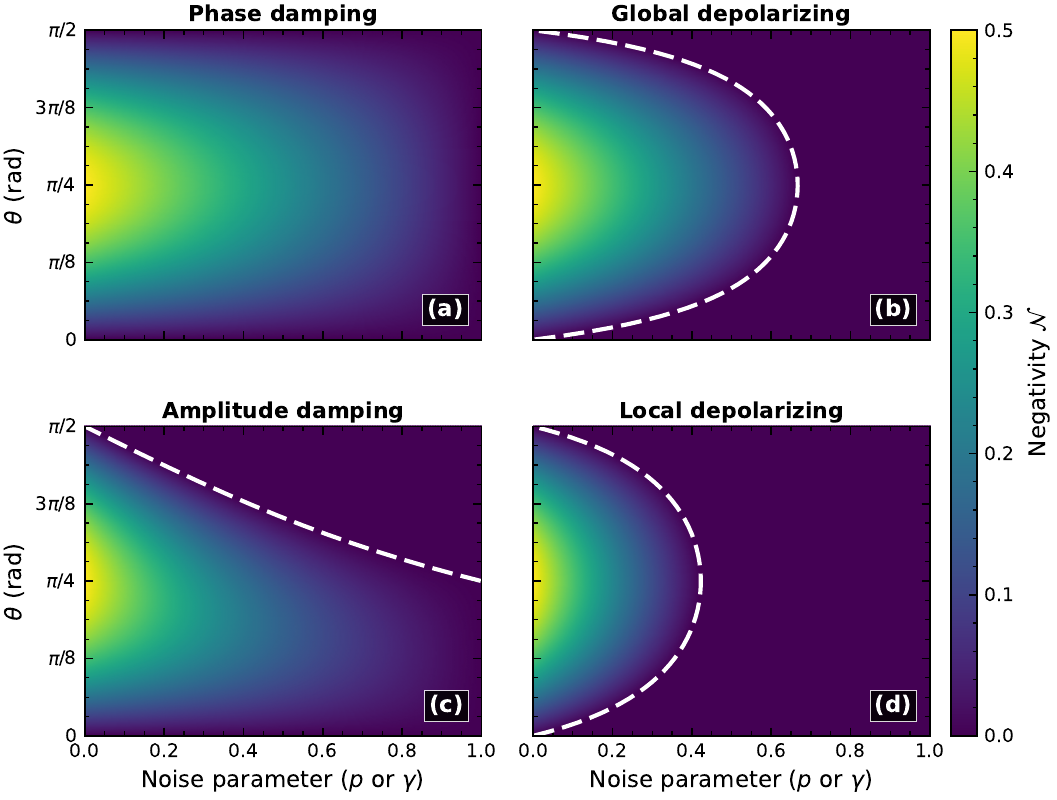}
  \caption{Two-dimensional negativity landscapes
    $\mathcal{N}(\theta,\cdot)$ over the full $(\theta,\,\text{noise})$ plane
    for all four channels on a shared colour scale:
    (a) phase damping, (b) global depolarizing, (c) amplitude damping,
    (d) local single-qubit depolarizing.
    White dashed curves mark the analytic entanglement-sudden-death boundary
    where one exists:
    $p_c(\theta)=2\sin(2\theta)/[2\sin(2\theta)+1]$ for global depolarizing,
    $\gamma_c(\theta)=\cot\theta$ (physical for $\theta>\pi/4$) for amplitude
    damping, and $p_c^{\mathrm{loc}}(\theta)=1-1/\!\sqrt{2\sin(2\theta)+1}$ for
    local depolarizing.
    Phase damping (a) has no boundary: the negativity decays multiplicatively
    as $(1-p)$ and never vanishes for $p<1$.
    The local-depolarizing ESD region (d) is the largest, reflecting that local
    noise is the most damaging per unit parameter.}
  \label{fig:landscape}
\end{figure*}

\begin{figure*}[t]
  \centering
  \includegraphics[width=0.96\textwidth]{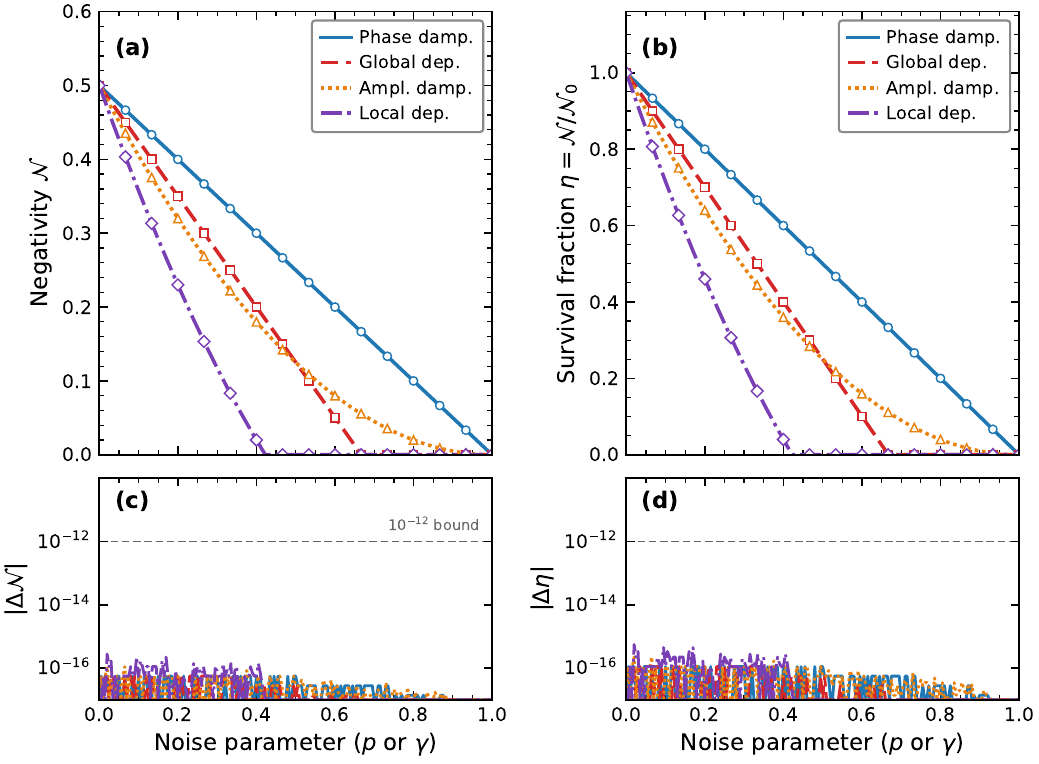}
  \caption{Consolidated head-to-head comparison of all four channels at
    $\theta=\pi/4$, plotted directly against the physical noise parameter
    ($p$ for phase, global, and local depolarizing; $\gamma$ for amplitude
    damping).
    (a) Negativity $\mathcal{N}$ and (b) entanglement-survival fraction
    $\eta=\mathcal{N}/\mathcal{N}_0$.
    All four channels start from $\mathcal{N}_0=\tfrac{1}{2}$ ($\eta=1$) and
    order by fragility: local depolarizing decays fastest and reaches sudden
    death first at $p_c^{\mathrm{loc}}\approx0.42$, followed by global
    depolarizing at $p_c=2/3$; phase damping decays linearly, reaching zero only
    at $p=1$; amplitude damping retains entanglement up to $\gamma=1$ at this
    angle.
    Curves are the closed-form expressions and open symbols the direct
    density-matrix simulation.
    (c,d) Absolute numerical deviations $|\Delta\mathcal{N}|$ and $|\Delta\eta|$
    between simulation and analytics on a logarithmic scale; for every channel
    these remain at the $10^{-16}$ level of double-precision roundoff, far below
    the quoted bound $\Delta_{\max}<10^{-12}$ (dashed reference line).
    Each channel keeps a fixed colour, line style, and marker throughout the
    paper.}
  \label{fig:three_channel}
\end{figure*}

\begin{figure*}[t]
  \centering
  \includegraphics[width=0.96\textwidth]{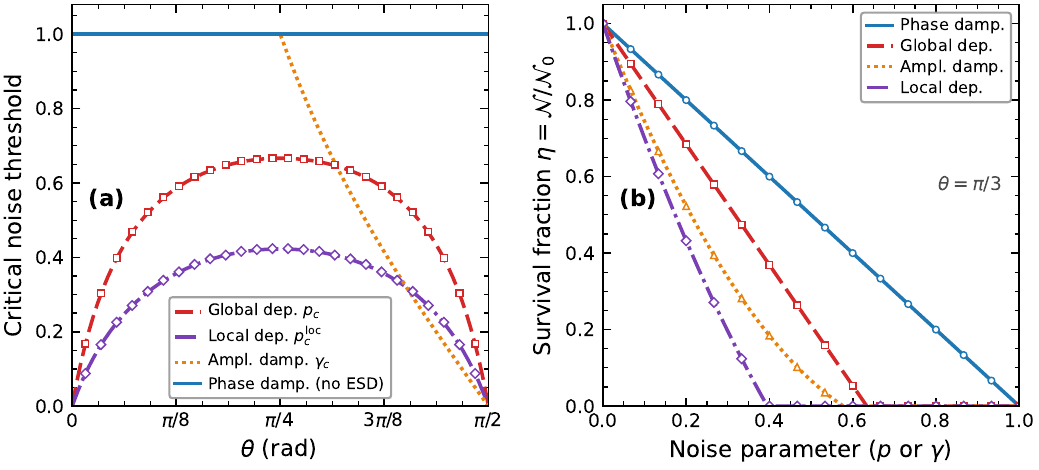}
  \caption{(a) Entanglement-sudden-death thresholds versus input angle $\theta$
    for all four channels on one axis.  Global depolarizing $p_c(\theta)$ and
    local depolarizing $p_c^{\mathrm{loc}}(\theta)$ both exist for every
    $\theta\in(0,\pi/2)$, with $p_c^{\mathrm{loc}}<p_c$ everywhere because
    $(1-p)^2<(1-p)$; amplitude damping $\gamma_c(\theta)=\cot\theta$ is physical
    only for $\theta>\pi/4$; phase damping has no finite threshold (flat line at
    the top, ``no ESD'').
    (b) Angle-resolved survival fraction $\eta=\mathcal{N}/\mathcal{N}_0$ versus noise
    parameter at the representative excited input $\theta=\pi/3$, where all four
    channels (including amplitude damping) exhibit sudden death; the ordering
    matches panel (a).
    Open symbols: numerically extracted values ($\Delta_{\max}<10^{-12}$).}
  \label{fig:ESD}
\end{figure*}

\section{Derivation of the $T_1$, $T_2$, and Fidelity Mappings}
\label{app:device}

This appendix derives the physical-parameter relations quoted in
Section~\ref{sec:physical_params}.

\textbf{Amplitude damping and $T_1$.}
Energy relaxation is governed by the Lindblad equation
$\dot\rho=\Gamma_1\bigl(\sigma_-\rho\sigma_+-\tfrac{1}{2}\{\sigma_+\sigma_-,\rho\}\bigr)$
with $\Gamma_1=1/T_1$ and $\sigma_-=\ket{0}\!\bra{1}$.  The excited-state
population obeys $\dot\rho_{11}=-\Gamma_1\rho_{11}$, so
$\rho_{11}(t)=\rho_{11}(0)\,e^{-t/T_1}$ and a fraction
$1-e^{-t/T_1}$ of any excitation decays in time $t$.  Identifying this decayed
fraction with the amplitude-damping Kraus parameter gives
\begin{equation}
  \gamma=1-e^{-t/T_1}.
\end{equation}
The coherences decay as $\rho_{01}(t)=\rho_{01}(0)\,e^{-t/2T_1}$, i.e.\ relaxation
alone contributes $1/(2T_1)$ to the total dephasing rate.

\textbf{Phase damping, $T_\varphi$, and $T_2$.}
Pure dephasing is described by
$\dot\rho=\tfrac{\Gamma_\varphi}{2}\bigl(\sigma_z\rho\sigma_z-\rho\bigr)$
with $\Gamma_\varphi=1/T_\varphi$, under which the off-diagonal element decays as
$\rho_{01}(t)=\rho_{01}(0)\,e^{-t/T_\varphi}$ while populations are unchanged.
The phase-damping Kraus operators of Appendix~\ref{app:kraus} suppress the
single-qubit coherence by the factor $\sqrt{1-p}$, so matching the Kraus
suppression to the physical coherence decay, $\sqrt{1-p}=e^{-t/T_\varphi}$, fixes
the noise parameter as
\begin{equation}
  p=1-e^{-2t/T_\varphi}.
  \label{eq:p_Tphi}
\end{equation}
(The factor of two in the exponent is a direct consequence of the Kraus
parametrization, in which the coherence is multiplied by $\sqrt{1-p}$ rather than
by $1-p$; with the alternative symmetric parametrization
$K_0=\sqrt{1-\tilde p/2}\,\mathbb{I}$, $K_1=\sqrt{\tilde p/2}\,\sigma_z$ one would
instead obtain $\tilde p=1-e^{-t/T_\varphi}$~\cite{krantz2019quantum}.  The two
conventions describe the same physical channel and the negativity results are
unchanged, since they depend only on the density-matrix coherence factor.)
Combining the relaxation contribution $1/(2T_1)$ with the pure-dephasing
contribution $1/T_\varphi$ gives the standard total-dephasing (transverse
relaxation) rate~\cite{krantz2019quantum}
\begin{equation}
  \frac{1}{T_2}=\frac{1}{2T_1}+\frac{1}{T_\varphi}.
\end{equation}

\textbf{Depolarizing and average gate fidelity.}
For a depolarizing channel on a $d$-dimensional system,
$\mathcal{E}(\rho)=(1-p)\rho+p\,\mathbb{I}_d/d$, the average gate fidelity with the
identity, obtained by integrating over the Haar
measure~\cite{nielsen2002simple}, is
\begin{equation}
  F_{\mathrm{avg}}=\int\!\mathrm{d}\psi\,
  \bra{\psi}\mathcal{E}(\ket{\psi}\!\bra{\psi})\ket{\psi}
  =1-p\,\frac{d-1}{d}.
\end{equation}
Solving for $p$ gives the average-gate-fidelity relation
\begin{equation}
  p=\frac{d}{d-1}\bigl(1-F_{\mathrm{avg}}\bigr)
   =\frac{4}{3}\bigl(1-F_{\mathrm{avg}}\bigr)\qquad(d=4).
  \label{eq:p_Favg}
\end{equation}
A different but equally common convention reports the depolarizing parameter
relative to the $d^2-1$ non-identity Pauli channels, as extracted directly from
randomised benchmarking~\cite{magesan2011scalable}; in that convention
\begin{equation}
\begin{split}
  p_{\mathrm{RB}}&=\frac{d^{2}}{d^{2}-1}\bigl(1-F_{\mathrm{avg}}\bigr)\\
   &=\frac{16}{15}\bigl(1-F_{\mathrm{avg}}\bigr)
    =\frac{16}{15}\,\varepsilon_{\mathrm{gate}} \quad(d=4).
\end{split}
  \label{eq:p_RB}
\end{equation}
Equations~\eqref{eq:p_Favg} and~\eqref{eq:p_RB} differ because they normalise the
error to the average gate fidelity and to the randomised-benchmarking Pauli rate
respectively.  The parameter appearing in our negativity formulas is the $p$ of
$\mathcal{E}_{\mathrm{dep}}(\rho)=(1-p)\rho+p\,\mathbb{I}_4/4$, whose
self-consistent fidelity relation is Eq.~\eqref{eq:p_Favg}, $p=\tfrac{4}{3}
(1-F_{\mathrm{avg}})$; this is the form quoted in Eq.~\eqref{eq:p_gate}.  The
randomised-benchmarking parameter $p_{\mathrm{RB}}$ of Eq.~\eqref{eq:p_RB} is a
distinct quantity and is not interchangeable with it.

\textbf{Local depolarizing and the per-qubit Pauli rate.}
A single-qubit depolarizing channel written in terms of an isotropic Pauli-error
probability $q$ (equal probability $q/3$ for each of $X,Y,Z$) matches the form in
Appendix~\ref{app:kraus} under $p/4=q/3$ for each non-identity Pauli weight, i.e.
\begin{equation}
  p=\frac{4q}{3}.
\end{equation}
Thus $q$ is directly accessible from single-qubit randomised benchmarking, and the
local-depolarizing results of Section~\ref{sec:loc} apply with this substitution.

\section{Kraus Operators for All Noise Channels}
\label{app:kraus}

For completeness we collect the single-qubit Kraus sets for all channels in one
place and give the two-qubit construction used throughout.  Each single-qubit
channel $\mathcal{E}(\rho)=\sum_k K_k\rho K_k^\dagger$ satisfies the completeness
relation $\sum_k K_k^\dagger K_k=\mathbb{I}_2$.

\textbf{Phase damping} (dephasing strength $p$):
\begin{equation}
  K_0^{\mathrm{pd}}=\begin{pmatrix}1&0\\0&\sqrt{1-p}\end{pmatrix},\qquad
  K_1^{\mathrm{pd}}=\begin{pmatrix}0&0\\0&\sqrt{p}\end{pmatrix},
\end{equation}
with $K_0^\dagger K_0+K_1^\dagger K_1=\mathrm{diag}(1,(1-p)+p)=\mathbb{I}_2$.

\textbf{Amplitude damping} (decay probability $\gamma$):
\begin{equation}
  K_0^{\mathrm{ad}}=\begin{pmatrix}1&0\\0&\sqrt{1-\gamma}\end{pmatrix},\qquad
  K_1^{\mathrm{ad}}=\begin{pmatrix}0&\sqrt{\gamma}\\0&0\end{pmatrix},
\end{equation}
with $K_0^\dagger K_0+K_1^\dagger K_1=\mathrm{diag}(1,(1-\gamma)+\gamma)=\mathbb{I}_2$.

\textbf{Single-qubit depolarizing} (per-qubit probability $p$), used to build the
\emph{local} channel of Section~\ref{sec:loc}:
\begin{equation}
\begin{split}
  &K_0^{\mathrm{dep}}=\sqrt{1-\tfrac{3p}{4}}\,\mathbb{I}_2,\\
  &K_1=\sqrt{\tfrac{p}{4}}\,X,\quad
  K_2=\sqrt{\tfrac{p}{4}}\,Y,\quad
  K_3=\sqrt{\tfrac{p}{4}}\,Z,
\end{split}
\end{equation}
where $X,Y,Z$ are the Pauli matrices.  Completeness follows from
$X^\dagger X=Y^\dagger Y=Z^\dagger Z=\mathbb{I}_2$:
$\sum_{k=0}^{3}K_k^\dagger K_k=(1-\tfrac{3p}{4})\mathbb{I}_2+3\cdot\tfrac{p}{4}\mathbb{I}_2
=\mathbb{I}_2$.

\textbf{Two-qubit construction.}
For the three Kraus channels above (phase damping, amplitude damping, local
depolarizing) the post-gate map is the independent product
$\mathcal{E}^{(2)}=\mathcal{E}\otimes\mathcal{E}$, with two-qubit Kraus operators
\begin{equation}
\begin{split}
  K_{ij}&=K_i\otimes K_j,\\
  \sum_{ij}K_{ij}^\dagger K_{ij}
  &=\Bigl(\sum_i K_i^\dagger K_i\Bigr)\otimes\Bigl(\sum_j K_j^\dagger K_j\Bigr)
  =\mathbb{I}_4 .
\end{split}
\end{equation}
Phase damping and amplitude damping each contribute $2\times2=4$ two-qubit Kraus
operators; local depolarizing contributes $4\times4=16$.

\textbf{Global depolarizing.}
The global two-qubit depolarizing channel is \emph{not} a product of single-qubit
maps.  In operator-sum form it admixes the maximally mixed state directly,
\begin{equation}
  \mathcal{E}_{\mathrm{dep}}(\rho)=(1-p)\rho+\frac{p}{4}\,\mathbb{I}_4
  =(1-p)\rho+\frac{p}{16}\sum_{\mu=0}^{15}P_\mu\,\rho\,P_\mu ,
\end{equation}
where $\{P_\mu\}=\{\mathbb{I},X,Y,Z\}^{\otimes2}$ are the sixteen two-qubit Pauli
operators; the second form makes the trace-preservation
$\tfrac{1}{16}\sum_\mu P_\mu\rho P_\mu=\tfrac{1}{4}\mathbb{I}_4\,\mathrm{Tr}\rho$
manifest.

\bibliography{bibliography}
\end{document}